\documentclass[submission,copyright,creativecommons]{eptcs}
\providecommand{\event}{SOS 2007} 
\usepackage{breakurl}             
\usepackage{underscore}           
\usepackage{latexsym}
\usepackage{amssymb}
\usepackage{amsmath,amsthm,bussproofs,units,mathrsfs,enumerate}
\newcommand{\true}{{\tt T}}
\newcommand{\both}{{\tt B}}
\newcommand{\none}{{\tt N}}
\newcommand{\false}{{\tt F}}
\newcommand{\ND}{\mathfrak{ND}}
\newcommand{\Sfde}{{\bf S_{fde}}}
\newcommand{\Srfde}{{\bf S_{fde}^\rightarrow}}
\newcommand{\Slfde}{{\bf S_{fde}^\leftarrow}}
\newcommand{\dSfde}{{\bf dS_{fde}}}
\newcommand{\dSrfde}{{\bf dS_{fde}^\rightarrow}}
\newcommand{\dSlfde}{{\bf dS_{fde}^\leftarrow}}
\theoremstyle{definition}
\newtheorem{defn}{\sc Definition}[section]
\newtheorem{theorem}{\sc Theorem}[section]
\newtheorem{lemma}{\sc Lemma}[section]
\newtheorem{proposition}{\sc Proposition}[section]
\title{Normalisation for Some Infectious Logics and Their Relatives}
\author{Yaroslav Petrukhin
\institute{University of \L{}\'{o}d\'{z}\\ \L{}\'{o}d\'{z}, Poland}
\email{yaroslav.petrukhin@mail.ru}}
\def\titlerunning{Normalisation for Some Infectious Logics \ldots}
\def\authorrunning{Y. Petrukhin}
\begin{document}
\maketitle

\begin{abstract}
We consider certain infectious logics ($ \Sfde $, $ \dSfde $, $ \bf K_3^w$, and \textbf{PWK}) and several their non-infectious modifications, including two new logics, reformulate previously constructed natural deduction systems for them (or present such systems from scratch for the case of new logics) in way such that the proof of normalisation theorem becomes possible for these logics. We present such a proof and establish the negation subformula property for the logics in question.
\end{abstract}

\section{Introduction}
Although the term `infectious logic' is relatively new \cite{Ferguson15}, the first representative of this direction in logic is the weak Kleene logic $ {\bf K_3^w} $ \cite{Kleene38} which is a fragment of Bochvar's logic $ \bf B_3$ \cite{Bochvar}. Kleene's motivation for the introduction of his logic was connected with the recursion theory and ordinal numbers, while  Bochvar's  motivation was the development of the logical instrument for the analysis of semantic paradoxes, mainly Russell's paradox. We may say that infectious logic is a part of a wider field of logic called nonsense logic started from Bochvar's paper \cite{Bochvar}, continued by Halld\'{e}n's monograph \cite{Hallden} (where the logic \textbf{PWK} was introduced, i.e. $ {\bf K_3^w} $ with two designated values) and papers by various authors such as \AA{}qvist \cite{Aqvist}, Ebbinghaus \cite{Ebbinghaus}, Finn and Grigolia \cite{FinnGrigolia}, Ha\l{}kowska \cite{Halkowska}. It is not the case that each nonsense logic is an infectious logic (while every infectious logic is a nonsense one), e.g. Ha\l{}kowska's nonsense logic \textbf{Z} is not an infectious logic. We say that a logic is infectious, if it has an infectious value, i.e. a value such that if one of compounds of a formula is evaluated by it, then the whole formula is evaluated by it as well. It is not the case that all nonsense logics have such a value. This is not their drawback, but in the recent literature there is a special interest for the logics which have an infectious value. Let us mention some works in this field. Szmuc \cite{Szmuc} studied the connection of infectious logics with logics of formal inconsistency and 
undeterminedness, Ciuni, Szmuc, and Ferguson \cite{CiuniSzmucFerguson} explored the connection of infectious logics with relevant ones. Proof-theoretical investigation (mainly based on sequent calculi) of infectious logics has been carried out by various authors in \cite{Szmuc,CiuniSzmucFerguson,CiuniSzmucFerguson19,BelikovPetrukhin,Belikov,ConiglioCorbalan,Szmuc21,Fjellstad,Petrukhin3,Petrukhin4,Petrukhin}. Algebraic treatment of infectious logics is presented, e.g. in \cite{Bonzioetla,Baldi}. Epistemic interpretation of infectious logics is developed in \cite{Szmuc19}. Theories of truth based on infectious logics are formulated in \cite{DaReetla}. 
 For more references about infectious logics, see, e.g. \cite{BelikovPetrukhin,Belikov}. 

As was said in the abstract, we are going to consider $ \Sfde $, $ \dSfde $, $ \bf K_3^w$, \textbf{PWK}, and some other logics. As for  $ \bf K_3^w$ and \textbf{PWK}, we said above that they were introduced in \cite{Kleene38,Bochvar} and \cite{Hallden}, respectively. What about $ \Sfde $ and $ \dSfde $? $ \Sfde $ is Deutsch's logic \cite{Deutsch77} and the motivation of its investigation is connected with relevant logic. Later on it was independently discovered by Fitting \cite{Fitting94} in the context of the study of four-valued generalizations of Kleene's three-valued logics, bilattices, and their computer science applications, and by Oller \cite{Oller} in the context of examination of paraconsistency and analyticity. $ \dSfde $ was introduced by Szmuc \cite{Szmuc} during the investigation of the connection between infectious logics and logics of formal inconsistency and undeterminedness. 

Natural deduction (ND for short) systems for $ \bf K_3^w$ and \textbf{PWK} are offered in \cite{Petrukhin3} (and later in \cite{Belikov}), for $ \Sfde $ in \cite{Petrukhin} (where this logic is called $ \bf FDE^\leftrightarrow$; and then in \cite{Belikov}), and for $ \dSfde $ in \cite{Belikov}.
In these papers, soundness and completeness theorems  are proven, but the issue of normalisation has not been considered.
As we show in this paper, for ND systems from  \cite{Petrukhin,Petrukhin3} after some minor changes of their rules in the spirit of the paper \cite{KurbisPetrukhin} normalisation and the negation subformula property can be established, while some of the rules of the systems from \cite{Belikov} destroy any meaningful subformula property, block the proof of normalisation, and are neither introduction, nor elimination rules. 

Since in \cite{Petrukhin} not only $ \Sfde $ was formalized via an ND system, but two more logics, Fitting's $ \bf FDE^\rightarrow$ \cite{Fitting94} and the logic $ \bf FDE^\leftarrow$ introduced in \cite{Petrukhin}, which for the unification of notation we will call here $ \Srfde $ and $ \Slfde $, respectively, we show that our methods work for them as well, i.e. ND systems for them (after minor changes of the rules) are normalisable and have the negation subformula property. Fitting's $ \Srfde $ was motivated by computer science problems. As for $ \Slfde $, it appeared in the context of exploration of four-valued generalization of Kleene's three-valued logics.  Taking our inspiration from the logic $ \dSfde $, we introduce two new logics, $ \dSrfde $ and $ \dSlfde $, respectively, which are in the same relations with $ \Srfde $ and $ \Slfde $ as $ \dSfde $ with $ \Sfde $. ND systems (with normalisation and the negation subformula property) for $ \dSrfde $ and $ \dSlfde $ are introduced.

In \cite{Belikov} it is emphasized that $ \bf K_3^w$ and \textbf{PWK} can be formalized as extensions of $ \Sfde $ and $ \dSfde $, respectively, by ex contradictione quodlibet and the law of excluded middle. 
  We demonstrate that just like $ \bf K_3^w$ extends $ \Sfde $ by ex contradictione quodlibet, McCarthy's $ \bf K_3^\rightarrow$ \cite{McCarthy} (independently reopened by Fitting \cite{Fitting94}) and Komendantskaya's \cite{Komendantskaya} $ \bf K_3^\leftarrow$ extend $ \Srfde $ and $ \Slfde$, respectively, by ex contradictione quodlibet; just like \textbf{PWK} extends $ \dSfde $ by the law of excluded middle, the logics  $ \bf K_3^{\rightarrow 2}$ and $ \bf K_3^{\leftarrow 2}$ introduced in \cite{Petrukhin3} extend  $ \dSrfde $ and $ \dSlfde$, respectively, by excluded middle. In fact, ND systems for $ \bf K_3^\rightarrow$, $ \bf K_3^\leftarrow$, $ \bf K_3^{\rightarrow 2}$, and $ \bf K_3^{\leftarrow 2}$ were first formulated in \cite{Petrukhin3}, but without connection with $ \Srfde $, $ \Slfde$, $ \dSrfde $, and $ \dSlfde$, and normalisation was not proven for them there. Here we fill this gap. Notice that in \cite{Kleene38} Kleene introduced the notion of a regular logic. As follows from \cite{Komendantskaya}, in the case of three-valued logics with one designated value  $ \bf K_3$,  $ \bf K_3^\rightarrow$, $ \bf K_3^\leftarrow$, and $ \bf K_3^w$ are all regular logics. In the case of three-valued logics with two designated values, \textbf{LP},  $ \bf K_3^{\rightarrow 2}$, $ \bf K_3^{\leftarrow 2}$, and $ \bf PWK$ are all regular logics.
  


The structure of the paper is as follows.  Section \ref{2} describes semantics for $ \Sfde $ and $ \dSfde $ and formally explains what is infectious logic. ND systems for $ \Sfde $ from \cite{Petrukhin} and \cite{Belikov} are compared. A modification of the system for $ \Sfde $ from \cite{Petrukhin} as well as a new ND system for $ \dSfde $ are presented. Normalisation for these new systems is proven and the negation subformula property is established. 
 Section \ref{3} is devoted to  the consideration of $ \Srfde $ and $ \Slfde $ as well as $ \dSrfde $ and $ \dSlfde $. Section \ref{5} contains completeness proof for ND systems for $ \dSfde $, $ \dSrfde $, and $ \dSlfde $. Section \ref{4} is devoted to the study of the logics $ \bf K_3^\rightarrow$, $ \bf K_3^\leftarrow$, $ \bf K_3^w$, $ \bf K_3^{\rightarrow 2}$, $ \bf K_3^{\leftarrow 2}$,  and \textbf{PWK}. Section \ref{6} contains concluding remarks.

\section{Semantics, natural deduction and normalisation for $ \Sfde $ and $ \dSfde $}\label{2}
\begin{defn}
	A logical matrix $ \langle \mathscr{V},\mathscr{C},\mathscr{D}\rangle $ (where $ \mathscr{V} $ is the set of truth values, $ \mathscr{C} $ is the set of connectives, $ \mathscr{D} $ is the set of designated values) has an infectious value $ i $ iff for each $ \circ\in \mathscr{C}$ and each $ \{x_1,\ldots,x_n\}\subseteq \mathscr{V}$ (where $ n>0 $) it holds that if $ i\in\{x_1,\ldots,x_n\} $, then
	$ \circ(x_1,\ldots,x_n)=i $.
\end{defn}
\begin{defn}
	A logic having a logical matrix with an infectious value is said to be infectious.
\end{defn}

Let us consider the standard propositional language with the connectives $ \neg $, $ \wedge $, and $ \vee $. The notion of a formula is defined in a standard way. 
Let $ \mathscr{V}=\{\true,\both,\none,\false\} $, where the values are understood in the Belnapian way \cite{Belnap1}: `true', `both true and false', `neither true, nor false', and `false'. 
Consider the following matrices presented below. We can see here a negation which is common for both $ \Sfde $ and $ \dSfde $, conjunction and disjunction for $ \Sfde $ (at the left) and conjunction and disjunction for $ \dSfde $ (at the right). We can observe that in $ \Sfde $ $ \none $ is an infectious value, while in $ \dSfde $ $ \both $ is an infectious value.
\begin{center}
	\begin{tabular}{|c|c|}
		\hline $ A $ & $ {\neg} $ \\ 
		\hline $ \true $ & $ \false $ \\ 
		\hline $ \both $ & $ \both $ \\ 
		\hline $ \none $ & $ \none $ \\ 
		\hline $ \false $ & $ \true $ \\ 
		\hline 
	\end{tabular} 
	\begin{tabular}{|c|cccc|}\hline
		$ \wedge $ & $\true$ & $\both$& $\none$& $\false$ \\\hline
		$\true$ & $\true$ & $\both$& $\none$& $\false$ \\\hline
		$\both$& $\both$& $\both$& $\none$& $\false$ \\\hline
		$\none$& $\none$& $\none$& $\none$& $\none$\\\hline
		$\false$ & $\false$ & $\false$ & $\none$& $\false$ \\\hline
	\end{tabular}
	\begin{tabular}{|c|cccc|}\hline
		$ \vee $ & $\true$ & $\both$ & $\none$ & $\false$ \\\hline
		$\true$ & $\true$ & $\true$ & $\none$ & $\true$ \\\hline
		$\both$ & $\true$ & $\both$ & $\none$ & $\both$ \\\hline
		$\none$ & $\none$ & $\none$ & $\none$ & $\none$\\\hline
		$\false$ & $\true$ & $\both$ & $\none$ & $\false$ \\\hline
	\end{tabular}
	\begin{tabular}{|c|cccc|}\hline
		$ \wedge $ & $\true$ & $\both$& $\none$& $\false$ \\\hline
		$\true$ & $\true$ & $\both$& $\none$& $\false$ \\\hline
		$\both$& $\both$& $\both$& $\both$& $\both$ \\\hline
		$\none$& $\none$& $\both$& $\none$& $\false$\\\hline
		$\false$ & $\false$ & $\both$ & $\false$& $\false$ \\\hline
	\end{tabular}
	\begin{tabular}{|c|cccc|}\hline
		$ \vee $ & $\true$ & $\both$ & $\none$ & $\false$ \\\hline
		$\true$ & $\true$ & $\both$ & $\true$ & $\true$ \\\hline
		$\both$ & $\both$ & $\both$ & $\both$ & $\both$ \\\hline
		$\none$ & $\true$ & $\both$ & $\none$ & $\none$\\\hline
		$\false$ & $\true$ & $\both$ & $\none$ & $\false$ \\\hline
	\end{tabular}  	  
\end{center}

The entailment relation is defined as follows. 
Let $ \bf L\in\{\Sfde,\dSfde\}$. Then:
\begin{center}
	$ \Gamma\models_{\bf L}\Delta $ iff $ v(A)\in \{\true,\both\}$, for each $ A\in\Gamma $, implies $ v(B)\in \{\true,\both\}$, for some $ B\in\Delta $, for each valuation $ v $.	
\end{center}

Let us present the ND system $ \ND_{\Sfde}^P $ for $ \Sfde $ from \cite{Petrukhin}. It has the following rules:

\begin{center}
	($\wedge I$) $ \dfrac{A\quad B}{A \wedge B} $ \quad
	($\wedge E_{1} $) $ \dfrac{A \wedge B}{A} $ \quad
	($\wedge E_{2} $) $ \dfrac{A \wedge B}{B} $ \quad
	($\overleftrightarrow{\vee I_{1}}$) $ \dfrac{A\wedge \neg B}{A \vee B} $ \quad
	($\overleftrightarrow{\vee I_{2}}$) $ \dfrac{\neg A\wedge B}{A \vee B} $ \quad
	($\overleftrightarrow{\vee I_{3}}$) $ \dfrac{A\wedge B}{A \vee B} $ \quad
\end{center}
\begin{center}
	($\overleftrightarrow{\vee E}$) $ \dfrac{\begin{matrix}
		& [A\wedge \neg B] & [\neg A\wedge B] & [A\wedge B]\\
		A \vee B &  C &  C & C
		\end{matrix}}{C} $\quad
	($\neg \neg I$) $ \dfrac{A}{\neg \neg A} $ \quad
	($\neg \neg E$) $ \dfrac{\neg \neg A}{A} $ \quad
	\end{center}
	\begin{center}
	($\neg \vee I$) $ \dfrac{\neg A\wedge \neg B}{\neg (A \vee B)} $ \quad
	($\neg \vee E$) $ \dfrac{\neg (A \vee B)}{\neg A\wedge \neg B} $ \quad
	($\neg \wedge I$) $ \dfrac{\neg A\vee \neg B}{\neg (A \wedge B)} $ \quad
	($\neg \wedge E$) $ \dfrac{\neg (A \wedge B)}{\neg A\vee \neg B} $
\end{center}
The notion of a deduction is defined in a standard Gentzen-Prawitz-style way. 	
If we replace ($\overleftrightarrow{\vee I_{1}}$), ($\overleftrightarrow{\vee I_{2}}$), ($\overleftrightarrow{\vee I_{3}}$), and ($\overleftrightarrow{\vee E}$) with the standard rules ($\vee I_{1} $),  ($\vee I_{2} $), and ($\vee E $), we get Priest's \cite{Priest02} natural deduction system $ \ND_{\bf FDE} $ for Belnap-Dunn's \cite{Belnap1,Dunn} \textbf{FDE}. 
\begin{center}
	($\vee I_{1} $) $ \dfrac{A}{A \vee B} $ \qquad
	($\vee I_{2} $) $ \dfrac{B}{A \vee B} $ \qquad
	($\vee E$) $ \dfrac{\begin{matrix}
		& [A] & [B] &\\
		A \vee B &  C &  C
		\end{matrix}}{C} $\quad
\end{center}	

Now let us present Belikov's ND system $ \ND_{\Sfde}^B $ for $ \Sfde $ from \cite{Belikov}. It is obtained from $ \ND_{\Sfde}^P $ by the replacement of the rule ($\overleftrightarrow{\vee E}$)  with ($\vee E$), $ (LEM_1) $, and $ (\vee C) $.

\begin{center}
	$ (LEM_1) $ 	$ \dfrac{A\vee B}{A\vee\neg A} $\qquad
	$ (\vee C) $ 	$ \dfrac{A\vee B}{B\vee A} $\quad
\end{center} 	

Belikov \cite{Belikov} writes that his system is better than $ \ND_{\Sfde}^P $, mainly because it has a standard disjunction elimination rule instead of ($\overleftrightarrow{\vee E}$). We think that it is rather a debatable question, 
 because normalisation and subformula property are at the first place  
 among the criteria of a good ND system (see, e.g. \cite{StructuralProofTheory}). Both in \cite{Petrukhin} and \cite{Belikov} this issue was not considered, but the system from \cite{Petrukhin} after minor modifications is normalisable and has the negation subformula property, while the system from \cite{Belikov} has not. The problem is with the rule $ (\vee C) $. First of all, it is out from Gentzen's classification of rules: it is neither introduction, nor elimination rule. Second, it destroys any meaningful subformula property: it is possible to find a deduction such that $ A\vee B$ is the subformula of the conclusion or of assumptions of this deduction, while $ B\vee A$ does not (the simplest example of such deduction is $ A\vee B \vdash B\vee\neg B$; at that if we add this principle as a new rule to Belikov's system, it seems that $ (\vee C) $ does not become derivable). Third, this rule makes the proof search more complicated: it is not clear when it should be applied. The rule $ (LEM_1) $ is also out of Gentzen's classification. 
Let us present a modification of the system from \cite{Petrukhin} which we call $ \ND_{\Sfde}^\prime $. It has the rules ($\wedge I$), ($\wedge E_{1} $), ($\wedge E_{2} $), ($\neg \neg I$), ($\neg \neg E$), and the subsequent ones:

\begin{center}
	($\overleftrightarrow{\vee E^\prime}$) $ \dfrac{\begin{matrix}
		& [A, \neg B] & [\neg A, B] & [A, B]\\
		A \vee B &  C &  C & C
		\end{matrix}}{C} $\quad
	  	($\overleftrightarrow{\neg \wedge E^\prime}$)$ \dfrac{\begin{matrix}
	  		& [\neg A, B] & [A,\neg B] & [\neg A,\neg B]\\
	  		\neg(A \wedge B) &  C &  C & C
	  		\end{matrix}}{C} $\quad		
\end{center}
\begin{center}
		($\overleftrightarrow{\vee I_{1}^\prime}$) $ \dfrac{A\quad \neg B}{A \vee B} $ \quad
		($\overleftrightarrow{\vee I_{2}^\prime}$) $ \dfrac{\neg A\quad B}{A \vee B} $ \quad
		($\overleftrightarrow{\vee I_{3}^\prime}$) $ \dfrac{A\quad B}{A \vee B} $ \quad
 	($\neg \vee I^\prime$) $ \dfrac{\neg A\quad \neg B}{\neg (A \vee B)} $ \quad
($\neg \vee E^\prime_1$) $ \dfrac{\neg (A \vee B)}{\neg A} $ \quad
\end{center}
\begin{center}
		($\neg \vee E^\prime_2$) $ \dfrac{\neg (A \vee B)}{\neg B} $ 	\quad
($\overleftrightarrow{\neg \wedge I^\prime_1}$) $ \dfrac{\neg A\quad B}{\neg (A \wedge B)}$\quad
		($\overleftrightarrow{\neg \wedge I^\prime_2}$) $ \dfrac{A\quad\neg B}{\neg (A \wedge B)}$ \quad
		($\overleftrightarrow{\neg \wedge I^\prime_3}$) $ \dfrac{\neg A \quad\neg B}{\neg (A \wedge B)}$	
\end{center}

Let us compare this system with $ \ND_{\Sfde}^P $. 	
We deleted conjunctions in the premises of disjunction introduction rules and in the assumptions of the disjunction elimination rule (the notation $\begin{matrix}
[\neg A, B]\\
C
\end{matrix}$ means that the formula $ C $ is derivable from two assumptions, $ \neg A$ and B). We reformulated the rules for the negated disjunction (in the same way as it was done in \cite{KurbisPetrukhin} for the case of \textbf{LP}, \textbf{FDE}, and related logics). In a similar fashion we modified the rules for the negated conjunction. If we replace the rules ($\overleftrightarrow{\vee I_{1}^\prime}$), ($\overleftrightarrow{\vee I_{2}^\prime}$), ($\overleftrightarrow{\vee I_{3}^\prime}$), ($\overleftrightarrow{\vee E^\prime}$), 
($\overleftrightarrow{\neg \wedge I^\prime_1}$), 
($\overleftrightarrow{\neg \wedge I^\prime_2}$), 
($\overleftrightarrow{\neg \wedge I^\prime_3}$), 
($\overleftrightarrow{\neg \wedge E^\prime}$) with $ (\vee I_1) $, $ (\vee I_2) $, $ (\vee E) $ as well as with the presented below rules 
($ \neg\wedge I_1$), $ (\neg\wedge I_2) $, and $ (\neg\wedge E^\prime) $, we get a ND system for \textbf{FDE} introduced in \cite{KurbisPetrukhin}.
\begin{center}
	$ (\neg\wedge I_1) $ $ \dfrac{\neg A}{\neg(A\wedge B)} $\qquad
$ (\neg\wedge I_2) $ $ \dfrac{\neg B}{\neg(A\wedge B)} $\qquad	
($\neg \wedge E^\prime$)$ \dfrac{\begin{matrix}
	& [\neg A] & [\neg B] \\
	\neg(A \wedge B) &  C &  C 
	\end{matrix}}{C} $\quad
\end{center}

In \cite{KurbisPetrukhin} a detailed proof of normalisation for this system is presented. We will follow this proof and indicate the cases which are different for \textbf{FDE} and $\Sfde$.
Let us now, following \cite{KurbisPetrukhin}, recall the terminology regarding 
normalisation.

\begin{defn}
	A \emph{maximal formula} is an occurrence of a formula in a deduction that is the conclusion of an introduction rule and the major premise of an elimination rule.		
\end{defn}
\begin{defn}
	Rules of the kind of disjunction elimination are called \emph{del-rules}. 
\end{defn}  
\begin{defn}
	(a) A \emph{segment} is a sequence of two or more formula occurrences $C_1\ldots C_n$ in a deduction such that $C_1$ is not the conclusion of a del-rule, $C_n$ is not the minor premise of a del-rule and for every $i<n$, $C_i$ is the minor premise of a del-rule and $C_{i+1}$ its conclusion. 
	
	\noindent (b) The \emph{length} of a segment is the number of formulas occurrences of which it consists, its \emph{degree} is their degree. 
	
	\noindent (c) A segment is \emph{maximal} if and only if its last formula is the major premise of an elimination rule. 		
\end{defn} 	 

\begin{defn}
	The \emph{rank} of a deduction $\Pi$ is the pair $\langle d, l\rangle$, where $d$ is the highest degree of any maximal formula or maximal segment in $\Pi$, and $l$ is the sum of the number of maximal formulas and the sum of the lengths of all maximal segments in $\Pi$. If there are no maximal formulas or maximal segments in $\Pi$, $d$ and $l$ are both $0$.		
\end{defn}

Ranks are ordered lexicographically: $\langle d, l\rangle< \langle d', l'\rangle$ iff either $d<d'$, or $d=d'$ and $l<l'$. 

\begin{defn}
	A deduction is \emph{in normal form} if it contains neither maximal formulas nor maximal segments.		
\end{defn}	

\begin{defn}
	A deduction $ \Pi $ of a conclusion $ A $ from undischarged assumptions $ \Gamma $ satisfies the \emph{subformula property} iff every formula in the deduction is a subformula either of $ A $ or of a formula in $ \Gamma $.		
\end{defn}
\begin{defn}
	A deduction satisfies the \emph{negation subformula property} iff every formula occurrence in it is either a subformula of an undischarged assumption or of the conclusion or it is the negation of such a formula.	
\end{defn}

None of the ND systems considered in this paper has the subformula property, but we show that those of them which enjoy normalisation, have the negation subformula property. The proof of normalisation for $ \Sfde $ is similar for the proof for \textbf{FDE} from \cite{KurbisPetrukhin}, let us present those reductions which are different for these logics. 

Reduction procedures. 
%
Disjunction (1st case):
\begin{center}
	\AxiomC{$\Sigma_1$}
	\noLine
	\UnaryInfC{$A$}
	\AxiomC{$\Sigma_2$}
	\noLine
	\UnaryInfC{$\neg B$}	
	\BinaryInfC{$A\lor B$}
	\AxiomC{$[A,\neg B]$}
	\noLine
	\UnaryInfC{$\Pi_1$}
	\noLine
	\UnaryInfC{$C$}
	\AxiomC{$[\neg A,B]$}
	\noLine
	\UnaryInfC{$\Pi_2$}
	\noLine
	\UnaryInfC{$C$}
	\AxiomC{$[A,B]$}
	\noLine
	\UnaryInfC{$\Pi_3$}
	\noLine
	\UnaryInfC{$C$}
	\QuaternaryInfC{$C$}
	\noLine
	\UnaryInfC{$\Xi$}
	\DisplayProof \qquad $\leadsto$ \qquad
\AxiomC{$\Sigma_1$}
\noLine
\UnaryInfC{$A$}
\AxiomC{$\Sigma_2$}
\noLine
\UnaryInfC{$\neg B$}
\noLine
\BinaryInfC{$\Pi_1$}
\noLine
\UnaryInfC{$C$}
\noLine
\UnaryInfC{$\Xi$}
\DisplayProof
\end{center}

Disjunction (2nd case):

\begin{center}
	\AxiomC{$\Sigma_1$}
	\noLine
	\UnaryInfC{$\neg A$}
	\AxiomC{$\Sigma_2$}
	\noLine
	\UnaryInfC{$B$}	
	\BinaryInfC{$A\lor B$}
	\AxiomC{$[A,\neg B]$}
	\noLine
	\UnaryInfC{$\Pi_1$}
	\noLine
	\UnaryInfC{$C$}
	\AxiomC{$[\neg A,B]$}
	\noLine
	\UnaryInfC{$\Pi_2$}
	\noLine
	\UnaryInfC{$C$}
	\AxiomC{$[A,B]$}
	\noLine
	\UnaryInfC{$\Pi_3$}
	\noLine
	\UnaryInfC{$C$}
	\QuaternaryInfC{$C$}
	\noLine
	\UnaryInfC{$\Xi$}
	\DisplayProof \qquad $\leadsto$ \qquad
\AxiomC{$\Sigma_1$}
\noLine
\UnaryInfC{$\neg A$}
\AxiomC{$\Sigma_2$}
\noLine
\UnaryInfC{$B$}
\noLine
\BinaryInfC{$\Pi_2$}
\noLine
\UnaryInfC{$C$}
\noLine
\UnaryInfC{$\Xi$}
\DisplayProof
\end{center}	

Disjunction (3rd case):

\begin{center}
	\AxiomC{$\Sigma_1$}
	\noLine
	\UnaryInfC{$A$}
	\AxiomC{$\Sigma_2$}
	\noLine
	\UnaryInfC{$B$}	
	\BinaryInfC{$A\lor B$}
	\AxiomC{$[A,\neg B]$}
	\noLine
	\UnaryInfC{$\Pi_1$}
	\noLine
	\UnaryInfC{$C$}
	\AxiomC{$[\neg A,B]$}
	\noLine
	\UnaryInfC{$\Pi_2$}
	\noLine
	\UnaryInfC{$C$}
	\AxiomC{$[A,B]$}
	\noLine
	\UnaryInfC{$\Pi_3$}
	\noLine
	\UnaryInfC{$C$}
	\QuaternaryInfC{$C$}
	\noLine
	\UnaryInfC{$\Xi$}
	\DisplayProof \qquad $\leadsto$ \qquad
\AxiomC{$\Sigma_1$}
\noLine
\UnaryInfC{$A$}
\AxiomC{$\Sigma_2$}
\noLine
\UnaryInfC{$ B$}
\noLine
\BinaryInfC{$\Pi_3$}
\noLine
\UnaryInfC{$C$}
\noLine
\UnaryInfC{$\Xi$}
\DisplayProof
\end{center}

Negated conjunction (one of three cases):
\begin{center}
	\AxiomC{$\Sigma_1$}
	\noLine
	\UnaryInfC{$A$}
	\AxiomC{$\Sigma_2$}
	\noLine
	\UnaryInfC{$\neg B$}	
	\BinaryInfC{$\neg(A\wedge B)$}
	\AxiomC{$[A,\neg B]$}
	\noLine
	\UnaryInfC{$\Pi_1$}
	\noLine
	\UnaryInfC{$C$}
	\AxiomC{$[\neg A,B]$}
	\noLine
	\UnaryInfC{$\Pi_2$}
	\noLine
	\UnaryInfC{$C$}
	\AxiomC{$[\neg A,\neg B]$}
	\noLine
	\UnaryInfC{$\Pi_3$}
	\noLine
	\UnaryInfC{$C$}
	\QuaternaryInfC{$C$}
	\noLine
	\UnaryInfC{$\Xi$}
	\DisplayProof \qquad $\leadsto$ \qquad
\AxiomC{$\Sigma_1$}
\noLine
\UnaryInfC{$A$}
\AxiomC{$\Sigma_2$}
\noLine
\UnaryInfC{$\neg B$}
\noLine
\BinaryInfC{$\Pi_1$}
\noLine
\UnaryInfC{$C$}
\noLine
\UnaryInfC{$\Xi$}
\DisplayProof
\end{center}

%

Permutation Conversions (two examples): 

{\small \begin{center}
	\AxiomC{$A\lor B$}
	\AxiomC{$[A,\neg B]$}
	\noLine
	\UnaryInfC{$\Pi_1$}
	\noLine
	\UnaryInfC{$\neg(C\vee D)$}
	\AxiomC{$[\neg A, B]$}
	\noLine
	\UnaryInfC{$\Pi_2$}
	\noLine
	\UnaryInfC{$\neg(C\vee D)$}
	\AxiomC{$[A, B]$}
	\noLine
	\UnaryInfC{$\Pi_3$}
	\noLine
	\UnaryInfC{$\neg(C\vee D)$}	
	\QuaternaryInfC{$\neg(C\vee D)$}
	\UnaryInfC{$\neg C$}
	\DisplayProof\quad	
	$\leadsto$	
	\AxiomC{$A\lor B$}
	\AxiomC{$[A,\neg B]$}
	\noLine
	\UnaryInfC{$\Pi_1$}
	\noLine
	\UnaryInfC{$\neg(C\vee D)$}
	\UnaryInfC{$\neg C$}
	\AxiomC{$[\neg A, B]$}
	\noLine
	\UnaryInfC{$\Pi_2$}
	\noLine
	\UnaryInfC{$\neg(C\vee D)$}
	\UnaryInfC{$\neg C$}
	\AxiomC{$[A,B]$}
	\noLine
	\UnaryInfC{$\Pi_3$}
	\noLine
	\UnaryInfC{$\neg(C\vee D)$}
	\UnaryInfC{$\neg C$}	
	\QuaternaryInfC{$\neg C$}
	\DisplayProof
\end{center}}

\begin{center}
	\AxiomC{$\neg (A\land B)$}
	\AxiomC{$[\neg A,B]$}
	\noLine
	\UnaryInfC{$\Pi_1$}
	\noLine
	\UnaryInfC{$\neg \neg C$}
	\AxiomC{$[A,\neg B]$}
	\noLine
	\UnaryInfC{$\Pi_2$}
	\noLine
	\UnaryInfC{$\neg \neg C$}
	\AxiomC{$[\neg A,\neg B]$}
	\noLine
	\UnaryInfC{$\Pi_3$}
	\noLine
	\UnaryInfC{$\neg \neg C$}	
	\QuaternaryInfC{$\neg \neg C$}
	\UnaryInfC{$C$}
	\DisplayProof	
	$\leadsto$	
	\AxiomC{$\neg (A\land B)$}
	\AxiomC{$[\neg A,B]$}
	\noLine
	\UnaryInfC{$\Pi_1$}
	\noLine
	\UnaryInfC{$\neg \neg C$}
	\UnaryInfC{$C$}
	\AxiomC{$[A,\neg B]$}
	\noLine
	\UnaryInfC{$\Pi_2$}
	\noLine
	\UnaryInfC{$\neg \neg C$}
	\UnaryInfC{$C$}
	\AxiomC{$[\neg A,\neg B]$}
	\noLine
	\UnaryInfC{$\Pi_3$}
	\noLine
	\UnaryInfC{$\neg \neg C$}
	\UnaryInfC{$C$}
	\QuaternaryInfC{$C$}
	\DisplayProof
\end{center}

\begin{theorem}
All deductions in $ \Sfde $ can be normalised. \end{theorem}
\begin{proof}
	By induction on the rank of deductions, using the reduction steps. Similarly to \cite[Theorem 1]{KurbisPetrukhin}.
\end{proof}
\begin{theorem} $ \Sfde $ has the negation subformula property.	
\end{theorem}
\begin{proof}
Similarly to \cite[Theorem 2]{KurbisPetrukhin}.
\end{proof}

Let us present Belikov's ND system $ \ND_{\dSfde}^B $ for $ \dSfde $. It has the rules $ (\wedge I) $, $ (\vee I_1) $, $ (\vee I_2) $, $ (\vee E) $, ($\neg \neg I$), ($\neg \neg E$), ($\neg \vee I$), ($\neg \vee E$), ($\neg \wedge I$), ($\neg \vee E$), and the following ones:

\begin{center}
	$ (\wedge I_2) $ $ \dfrac{A\quad \neg A}{A\wedge B} $\qquad
	$ (\wedge C) $ $ \dfrac{A\wedge B}{B\wedge A} $\qquad
	$ \overleftrightarrow{(\wedge E_1)} $ $ \dfrac{A\wedge B}{A\vee B} $\qquad
	$ \overleftrightarrow{(\wedge E_2)} $ $ \dfrac{A\wedge B}{\neg A\vee B} $\qquad
	$ \overleftrightarrow{(\wedge E_3)} $ $ \dfrac{A\wedge B}{A\vee \neg B} $	
\end{center}

This system has similar problems as Belikov's system for $ \Sfde $. This time the troublemaker is the rule $ (\wedge C) $. Let us formulate a new ND system for $ \dSfde $ which we call $ \ND_{\dSfde}^\prime $. It has the rules $ (\vee I_1) $, $ (\vee I_2) $, $ (\vee E) $, $ (\wedge I) $, $ (\wedge I_2) $, ($\neg \neg I$), ($\neg \neg E$), ($\neg \wedge I_1$), ($\neg \wedge I_2$), ($ \neg\vee I^\prime$), $(\neg\wedge E^\prime)$, and the following ones: 

\begin{center}
		$ (\wedge I_3) $ $ \dfrac{B\quad\neg B}{A\wedge B} $\qquad
		($ \neg\vee I_2$) $ \dfrac{A\quad\neg A}{\neg(A\vee B)} $\qquad
		($ \neg\vee I_3$) $ \dfrac{B\quad\neg B}{\neg(A\vee B)} $\qquad
\end{center}
\begin{center}	
$ \overleftrightarrow{(\wedge E)} $ $ \dfrac{\begin{matrix}
	& [A,B] & [A,\neg A] & [B,\neg B]\\
	A \wedge B &  C &  C & C
	\end{matrix}}{C} $\qquad
 $ \overleftrightarrow{(\neg\vee E)}$ $ \dfrac{\begin{matrix}
			& [\neg A,\neg B] & [A,\neg A] & [B,\neg B]\\
			\neg(A \vee B) &  C &  C & C
			\end{matrix}}{C} $\quad
\end{center}

\begin{theorem}
	For any formula $ A $ and any set of formulas $ \Gamma $, it holds that $ \Gamma\vdash A$ in $ \ND_{\dSfde}^B $ iff $ \Gamma\vdash A$ in $ \ND_{\dSfde}^\prime $.
\end{theorem}
\begin{proof}
	By induction on the length of the derivation. Left for the reader.
\end{proof}

Note that in Section \ref{5} we present the completeness proof for the logic $ \dSrfde $ which can be easily adapted for $ \dSfde $. 
We can do the reduction steps in a similar way as for $ \Sfde $ and can state the following theorems.

\begin{theorem}
	All deductions in $ \dSfde $ can be normalised. \end{theorem}
\begin{theorem} $ \dSfde $ has the negation subformula property.	
\end{theorem}			

\section{Fitting-style relatives of $ \Sfde $ and $ \dSfde $}\label{3}
In \cite{Petrukhin}, ND systems for two more logics, Fitting's \cite{Fitting94} $ \Srfde $ and $ \Slfde $ introduced in \cite{Petrukhin} ($ \bf FDE^\rightarrow$ and $ \bf FDE^\leftarrow$ in the notation of \cite{Petrukhin}) were formulated. Let us  present the matrices for their conjunctions and disjunctions (the negation is the same as in $ \Sfde $): on the left we see the pair of conjunction and disjunction for $ \Srfde $, on the right the pair for $ \Slfde $.
\begin{center}
	\begin{tabular}{|c|cccc|}\hline
		$ \wedge $ & \true & \both & \none & \false \\\hline
		\true & \true & \both & \none & \false \\\hline
		\both & \both & \both & \false & \false \\\hline
		\none & \none & \none & \none & \none \\\hline
		\false & \false & \false & \false & \false \\\hline
	\end{tabular}
	\begin{tabular}{|c|cccc|}\hline
		$ \vee $ & \true & \both & \none & \false \\\hline
		\true & \true & \true & \true & \true \\\hline
		\both & \true & \both & \true & \both \\\hline
		\none & \none & \none & \none & \none\\\hline
		\false & \true & \both & \none & \false \\\hline
	\end{tabular}
	\begin{tabular}{|c|cccc|}\hline
		$ \wedge $ & \true & \both & \none & \false \\\hline
		\true & \true & \both & \none & \false \\\hline
		\both & \both & \both & \none & \false \\\hline
		\none & \none & \false & \none & \false \\\hline
		\false & \false & \false & \none & \false \\\hline
	\end{tabular}
	\begin{tabular}{|c|cccc|}\hline
		$\vee$ & \true & \both & \none & \false \\\hline
		\true & \true & \true & \none & \true \\\hline
		\both & \true & \both & \none & \both \\\hline
		\none & \true & \true & \none & \none\\\hline
		\false & \true & \both & \none & \false \\\hline
	\end{tabular} 
\end{center}	

Following the analogy with $ \Sfde $ and $ \dSfde $, we define two new logics which we call $ \dSrfde$ and $\dSlfde $. The negation is the same as in $ \Sfde $,  on the left we present the pair of conjunction and disjunction for $ \dSrfde $, on the right the pair for $ \dSlfde $.

\begin{center}
	\begin{tabular}{|c|cccc|}\hline
		$ \wedge $ & \true & \both & \none & \false \\\hline
		\true & \true & \both & \none & \false \\\hline
		\both & \both & \both & \both & \both \\\hline
		\none & \none & \false & \none & \false \\\hline
		\false & \false & \false & \false & \false \\\hline
	\end{tabular}
	\begin{tabular}{|c|cccc|}\hline
		$ \vee $ & \true & \both & \none & \false \\\hline
		\true & \true & \true & \true & \true \\\hline
		\both & \both & \both & \both & \both \\\hline
		\none & \true & \true & \none & \none\\\hline
		\false & \true & \both & \none & \false \\\hline
	\end{tabular}
	\begin{tabular}{|c|cccc|}\hline
		$ \wedge $ & \true & \both & \none & \false \\\hline
		\true & \true & \both & \none & \false \\\hline
		\both & \both & \both & \false & \false \\\hline
		\none & \none & \both & \none & \false \\\hline
		\false & \false & \both & \false & \false \\\hline
	\end{tabular}
	\begin{tabular}{|c|cccc|}\hline
		$\vee$ & \true & \both & \none & \false \\\hline
		\true & \true & \both & \true & \true \\\hline
		\both & \true & \both & \false & \both \\\hline
		\none & \true & \both & \none & \none\\\hline
		\false & \true & \both & \none & \false \\\hline
	\end{tabular} 
\end{center}

The entailment relation in the logics in question is defined as follows ($ \bf L\in\{\Srfde,\Slfde,\dSrfde,\dSlfde \} $).
\begin{center}
	$ \Gamma\models_{\bf L}\Delta $ iff $ v(A)\in \{\true,\both\}$, for each $ A\in\Gamma $, implies $ v(B)\in \{\true,\both\}$, for some $ B\in\Delta $, for each valuation $ v $.	
\end{center}

ND system $ \ND_{\Srfde} $ for $ \Srfde $ presented in \cite{Petrukhin} is obtained from $ \ND_{\Sfde} $ by the replacement of the rules $ (\overleftrightarrow{\vee I_{1}}) $,  $ (\overleftrightarrow{\vee I_{3}}) $, and $ (\overleftrightarrow{\vee E}) $ with the rules $ (\vee I_1) $  and ($\overrightarrow{\vee E}$). We present a new ND system $ \ND_{\Srfde}^\prime $ for $ \Srfde $ which is obtained from $ \ND_{\Sfde}^\prime $ by the replacement of the rules $ (\overleftrightarrow{\vee I_{1}^\prime}) $,  $ (\overleftrightarrow{\vee I_{3}^\prime}) $, $ (\overleftrightarrow{\vee E}^\prime) $, ($\overleftrightarrow{\neg\wedge I_1}^\prime$), ($\overleftrightarrow{\neg\wedge I_3}^\prime$), and ($\overleftrightarrow{\neg\wedge E}^\prime$) with the rules $(\vee I_{1})$,   $(\overrightarrow{\vee E}^\prime)$, $ (\neg\wedge I_1) $, and ($\overrightarrow{\neg\wedge E}^\prime$).

\begin{center}
	($\overrightarrow{\vee E}$) $ \dfrac{\begin{matrix}
		& [A] & [\neg A\wedge B] \\
		A \vee B &  C &  C
		\end{matrix}}{C} $\qquad
	($\overrightarrow{\vee E}^\prime$) $ \dfrac{\begin{matrix}
		& [A] & [\neg A, B] \\
		A \vee B &  C &  C
		\end{matrix}}{C} $\qquad
	($\overrightarrow{\neg\wedge E}^\prime$) $ \dfrac{\begin{matrix}
		& [\neg A] & [A,\neg B] \\
		\neg(A \wedge B) &  C &  C
		\end{matrix}}{C} $\qquad		
\end{center}

ND system $ \ND_{\Slfde} $ for $ \Slfde $ presented in \cite{Petrukhin} is obtained from $ \ND_{\Sfde} $ by the replacement of the rules  $ (\overleftrightarrow{\vee I_{2}}) $, $ (\overleftrightarrow{\vee I_{3}}) $, and $ (\overleftrightarrow{\vee E}) $ with the rules  $ (\vee I_2) $, and ($\overleftarrow{\vee E}$). We present a new ND system $ \ND_{\Slfde}^\prime $ for $ \Slfde $ which is obtained from $ \ND_{\Sfde}^\prime $ by the replacement of the rules $ (\overleftrightarrow{\vee I_{2}^\prime})$,  $(\overleftrightarrow{\vee I_{3}^\prime})$, and $ (\overleftrightarrow{\vee E}^\prime) $, $ \overleftrightarrow{(\neg\wedge I_2^\prime)} $, $ \overleftrightarrow{(\neg\wedge I_3^\prime)} $, and $(\overleftrightarrow{\neg\wedge E}^\prime)$ with the rules  $(\vee I_{2})$, $(\overleftarrow{\vee E}^\prime)$, $ (\neg\wedge I_2) $, and $(\overleftarrow{\neg\wedge E}^\prime)$.
\begin{center}
	($\overleftarrow{\vee E}$) $ \dfrac{\begin{matrix}
		& [A\wedge\neg B] & [B]\\
		A \vee B &  C &  C
		\end{matrix}}{C} $\qquad					
	($\overleftarrow{\vee E}^\prime$) $ \dfrac{\begin{matrix}
		& [A,\neg B] & [B]\\
		A \vee B &  C &  C
		\end{matrix}}{C} $\qquad
	($\overleftarrow{\neg\wedge E}^\prime$) $ \dfrac{\begin{matrix}
		& [\neg A,B] & [B]\\
		\neg(A \wedge B) &  C &  C
		\end{matrix}}{C} $\qquad	
\end{center} 

The proof for $ \Sfde $ can be easily adapted for $ \Srfde $ and $ \Slfde $.
\begin{theorem}
	All deductions in $ \Srfde $ and $ \Slfde $ can be normalised. \end{theorem}
\begin{theorem} $ \Srfde $ and $ \Slfde $ have the negation subformula property.	
\end{theorem}

Let us introduce ND systems for our new logics, $ \dSrfde $ and $ \dSlfde $. ND system $ \ND_{\dSrfde} $ for $ \dSrfde $ is obtained from $ \ND_{\dSfde}^\prime $ by the replacement of the rules $ (\wedge I_3) $, $ \overleftrightarrow{(\wedge E)} $, $ (\neg\vee I_3) $, and $ \overleftrightarrow{(\neg\vee E)} $ with $ \overrightarrow{(\wedge E)} $ and $ \overrightarrow{(\neg\vee E)} $. There is also an alternative option: replace the rules $ (\wedge I_3) $, $ \overleftrightarrow{(\wedge E)} $, $ (\neg\vee I_3) $, and $ \overleftrightarrow{(\neg\vee E)} $ with $ \overrightarrow{(\wedge E)^\prime} $, $ (\wedge E_1) $, $ \overrightarrow{(\neg\vee E^\prime)} $, and $ (\neg\vee E_1^\prime) $. Normalisation holds for both options. 
\begin{center}
$ \overrightarrow{(\wedge E)} $ $ \dfrac{\begin{matrix}
	& [A,B] & [A,\neg A]\\
	A \wedge B &  C &  C
	\end{matrix}}{C} $\quad
$ \overrightarrow{(\neg\vee E)} $ $ \dfrac{\begin{matrix}
	& [\neg A,\neg B] & [A,\neg A]\\
	\neg(A \vee B) &  C &  C
	\end{matrix}}{C} $\qquad
\end{center}
	\begin{center}
$ \overrightarrow{(\wedge E^\prime)} $ $ \dfrac{\begin{matrix}
	& [\neg A] & [B]\\
	A \wedge B &  C &  C
	\end{matrix}}{C} $\qquad
$ \overrightarrow{(\neg\vee E^\prime)} $ $ \dfrac{\begin{matrix}
	& [A] & [\neg B]\\
	\neg(A \vee B) &  C &  C
	\end{matrix}}{C} $\qquad	
\end{center}

ND system $ \ND_{\dSlfde} $ for $ \dSlfde $ is obtained from $ \ND_{\dSfde}^\prime $ by the replacement of the rules $ (\wedge I_2) $, $ \overleftrightarrow{(\wedge E)} $, $ (\neg\vee I_2) $, and $ \overleftrightarrow{(\neg\vee E)} $ with $ \overleftarrow{(\wedge E)} $ and $ \overleftarrow{(\neg\vee E)} $. There is also an alternative option: replace the rules $ (\wedge I_2) $, $ \overleftrightarrow{(\wedge E)} $, $ (\neg\vee I_2) $, and $ \overleftrightarrow{(\neg\vee E)} $ with $ \overleftarrow{(\wedge E^\prime)} $, $ (\wedge E_2) $, $ \overleftarrow{(\neg\vee E^\prime)} $, and $ (\neg\vee E_2^\prime) $. Normalisation holds for both options. 
\begin{center}
	$ \overleftarrow{(\wedge E)} $ $ \dfrac{\begin{matrix}
		& [A,B] & [B,\neg B]\\
		A \wedge B &  C &  C
		\end{matrix}}{C} $\quad
	$ \overleftarrow{(\neg\vee E)} $ $ \dfrac{\begin{matrix}
		& [\neg A,\neg B] & [B,\neg B]\\
		\neg(A \vee B) &  C &  C
		\end{matrix}}{C} $\qquad
\end{center}
\begin{center}
	$ \overleftarrow{(\wedge E^\prime)} $ $ \dfrac{\begin{matrix}
		& [A] & [\neg B]\\
		A \wedge B &  C &  C
		\end{matrix}}{C} $\qquad
	$ \overleftarrow{(\neg\vee E^\prime)} $ $ \dfrac{\begin{matrix}
		& [\neg A] & [B]\\
		\neg(A \vee B) &  C &  C
		\end{matrix}}{C} $\qquad	
\end{center}
\begin{theorem}
	All deductions in $ \dSrfde $ and $ \dSlfde $ can be normalised. \end{theorem}
\begin{theorem} $ \dSrfde $ and $ \dSlfde $ have the negation subformula property.	
\end{theorem}

\section{Completeness for $ \dSfde $, $ \dSrfde $, and $ \dSlfde $}\label{5}
\begin{theorem}
	Let $ \bf L\in\{\dSfde,\dSrfde,\dSlfde \}$. For any set of formulas $ \Gamma $ and any formula $ A $, it holds that $ \Gamma\models_{\bf L} A$ iff $ \Gamma\vdash A$ in $ \ND_{\bf L} $.
\end{theorem}
\begin{proof}
The soundness part of this theorem is by the induction on the length of deduction (before that one should check that all the rules are sound which is a routine exercise). The completeness part is by the Henkin-style argument in the style of Kooi and Tamminga \cite{KooiTamminga}. As an example, we show a proof for $ \dSrfde $. 

\begin{defn}\label{theory}
	We say that a set of formulas $ \Gamma $ is a $ \dSrfde $-theory iff $ \Gamma $ is not equal to the set of all formulas, is closed under $ \vdash $ (i.e. for any formula $ A $, if $ \Gamma\vdash A$, then $ A\in\Gamma $), and has the following properties, for any formulas $ A $ and $ B $:
\begin{itemize}\itemsep=0pt
\item if $ A\vee B\in\Gamma$, then $ A\in\Gamma$ or $ B\in\Gamma $,
\item if $ A\wedge B\in\Gamma $, then $ A,B\in\Gamma $ or $ A,\neg A\in\Gamma$,
\item if $ \neg(A\wedge B)\in\Gamma$, then $ \neg A\in\Gamma$ or $ \neg B\in\Gamma $,
\item if $ \neg(A\vee B)\in\Gamma $, then $ \neg A,\neg B\in\Gamma $ or $ A,\neg A\in\Gamma$.	
\end{itemize}
\end{defn}

\begin{defn}\label{elementhood}
	For any set of formulas $ \Gamma $ and any formula $ A $, we define the notion of $ A $'s elementhood in $ \Gamma $ (we follows Kooi and Tamminga's terminology \cite{KooiTamminga}) as follows:
     $$
     e(A,\Gamma)=\left\{\begin{array}{cll}
     \none &\hbox{~iff~} A \not\in\Gamma,\neg A\not\in\Gamma;\\
     \false &\hbox{~iff~} A \not\in\Gamma,\neg A\in\Gamma;\\
     \true &\hbox{~iff~} A\in\Gamma, \neg A \not\in\Gamma;\\
     \both &\hbox{~iff~} A \in\Gamma,\neg A\in\Gamma.\\
     \end{array}\right.
     $$ 
\end{defn} 
     
\begin{lemma}\label{lemma1}
For any $ \dSrfde $-theory $ \Gamma $ and any formulas $ A $ and $ B $, it holds that\textup{:}
\begin{enumerate}\itemsep=0pt
\item $\neg e(A, \Gamma) = e(\neg A, \Gamma)$;
\item $e(A, \Gamma)\vee e(B, \Gamma) = e(A \vee B, \Gamma)$;
\item $e(A, \Gamma)\wedge e(B, \Gamma) = e(A \wedge B, \Gamma)$.	
\end{enumerate}	
\end{lemma}

\begin{proof}
1. See \cite[Theorem 3.5]{Petrukhin}.
	
2.  Assume that $ e(A,\Gamma)=\true $ and $ e(B,\Gamma)=\both $. Then, by Definition \ref{elementhood}, $ A\in\Gamma $, $\neg A\not\in\Gamma $, $ B\in\Gamma $, and $ \neg B\in\Gamma $. By the rule $ (\vee I_1) $, $ A\vee B\in\Gamma$. Suppose that  $ \neg(A\vee B)\in\Gamma$. Then, since $ \Gamma $ is a $ \dSrfde $-theory, $ \neg A,\neg B\in\Gamma $ or $ A,\neg A\in\Gamma$. Since $ \neg A\not\in\Gamma$, both conditions are not fulfilled and we obtain contradiction. Hence, $ \neg(A\vee B)\not\in\Gamma$. Thus, by Definition \ref{elementhood}, $ e(A\vee B,\Gamma)=\true $. Therefore, $ e(A \vee B, \Gamma)  = \true = \true\vee \both = e(A, \Gamma)\vee e(B, \Gamma)$.

Assume that $ e(A,\Gamma)=\both $ and $ e(B,\Gamma)\in\{\true,\both\} $. Then $ A\in\Gamma $ and $\neg A\in\Gamma $. By the rule $ (\vee I_1) $, $ A\vee B\in\Gamma$. By the rule $ (\neg\vee I_2) $, $ \neg(A\vee B)\in\Gamma$. Thus, $ e(A\vee B,\Gamma)=\both $.

Assume that $ e(A,\Gamma)=\false $ and $ e(B,\Gamma)=\none $. Then $ A\not\in\Gamma $, $\neg A\in\Gamma $, $ B\not\in\Gamma $, and $ \neg B\not\in\Gamma $. If $ A\vee B\in\Gamma$, then $ A\in\Gamma $ or $ B\in\Gamma $. Contradiction. $ A\vee B\not\in\Gamma$. If $ \neg(A\vee B)\in\Gamma$, then $ \neg A,\neg B\in\Gamma $ or $ A,\neg A\in\Gamma$. Since $ A\not\in\Gamma $ and $ \neg B\not\in\Gamma $, we get contradiction. Hence, $ \neg(A\vee B)\not\in\Gamma$. Thus, $ e(A\vee B,\Gamma)=\none $.

The other cases are considered similarly. 

3. Assume that $ e(A,\Gamma)=\true $ and $ e(B,\Gamma)=\both $. By $ (\wedge I) $, $ A\wedge B\in\Gamma$. By $ (\neg\wedge I_2) $, $ \neg(A\wedge B)\in\Gamma $. Thus, $ e(A\wedge B,\Gamma)=\both $. 

Assume that $ e(A,\Gamma)=\both $ and $ e(B,\Gamma)\in\{\true,\both\} $. Then $ A\in\Gamma $, $\neg A\in\Gamma $, and $ B\in\Gamma $.  By $ (\wedge I) $, $ A\wedge B\in\Gamma$. By the rule $ (\neg\wedge I_1) $, $ \neg(A\wedge B)\in\Gamma$. Thus, $ e(A\wedge B,\Gamma)=\both $.

Assume that $ e(A,\Gamma)=\false $ and $ e(B,\Gamma)=\none $. If $ A\wedge B\in\Gamma$, then $ A,B\in\Gamma $ or $ A,\neg A\in\Gamma$. Since $ A\not\in\Gamma $, we obtain that $ A\wedge B\not\in\Gamma$. By the rule $ (\neg\wedge I_1) $, $ \neg(A\wedge B)\in\Gamma $. Hence, $ e(A\wedge B,\Gamma)=\false $. 

The other cases are considered similarly.
\end{proof}

  \begin{lemma} \label{lemma2}
  	Let $ \Gamma $ be an arbitrary $ \bf\dSrfde $-theory and $ v_{\Gamma} $ be an arbitrary valuation such that for any propositional variable $ p ,$ $ v_{\Gamma}(p) = e(p, \Gamma) $. Then, for any formula $ A $, it holds that $ v_{\Gamma}(A) = e(A, \Gamma) $.
  \end{lemma}
  
  \begin{proof} By a structural induction on formula $ A $ using the Lemma \ref{lemma1}.
  \end{proof}

\begin{lemma}[Lindenbaum]\label{lemma3}
For any set of formulas $ \Gamma $ and any formula $ A $, if $ \Gamma\not\vdash_{\dSrfde}A $, then there is a $ \dSrfde $-theory $ \Delta $ such that $ \Gamma\subseteq\Delta $ and $ \Delta\not\vdash_{\dSrfde}A $.
\end{lemma} 
\begin{proof}
This can be  proven by the standard methods (see, e.g. \cite[Lemma 3.8]{KooiTamminga}).	
\end{proof}

Suppose that $ \Gamma\not\vdash_{\dSrfde}A $. By Lemma \ref{lemma3} this implies an existence of a $ \dSrfde $-theory $ \Delta $ such that $ \Gamma\subseteq\Delta $ and $ \Delta\not\vdash_{\dSrfde}A $.
By Lemma \ref{lemma2}, we obtain that $ e(B,\Delta)\in\{\true,\both \} $ for any $ B\in\Gamma $, 
  while $ e(A,\Delta)\not\in\{\true,\both \} $, i.e. $ \Gamma\not\models_{\dSrfde} A$.
\end{proof}
\section{Three-valued extensions of the four-valued logics in question}\label{4}

In \cite{Petrukhin3} ND systems for $ \bf K_3^\rightarrow$, $ \bf K_3^\leftarrow$, $ \bf K_3^w$, $ \bf K_3^{\rightarrow 2}$, $ \bf K_3^{\leftarrow 2}$,  and \textbf{PWK} are offered. We show that these logics can be formalised as extensions of $ \Srfde $, $ \Slfde $, $ \Sfde $, $ \dSrfde $, $ \dSlfde $,  and $ \dSfde $, respectively.  
As for the semantics for these logics, matrices $ \bf K_3^\rightarrow$, $ \bf K_3^\leftarrow$, $ \bf K_3^w$ are $ \{\true,\none,\false \} $-restrictions of the matrices for $ \Srfde $, $ \Slfde $, and $ \Sfde $, respectively. Matrices $ \bf K_3^{\rightarrow 2}$, $ \bf K_3^{\leftarrow 2}$, $ \bf PWK$ are $ \{\true,\both,\false \} $-restrictions of the matrices for $ \dSrfde $, $ \dSlfde $, and $ \dSfde $, respectively.  
Let us start with ND systems from \cite{Petrukhin3}. Consider the following rules:
\begin{center}
	(EFQ)	$ \dfrac{A\quad \neg A}{B} $\qquad
	(EM) $ \dfrac{\begin{matrix}
		[A] & [\neg A] \\
		B &  B 
		\end{matrix}}{B} $	\qquad
	$ (\overrightarrow{\neg\wedge I_2}) $	$ \dfrac{A\wedge \neg B}{\neg(A\wedge B)} $\qquad
	$ (\overleftarrow{\neg\wedge I_1}) $	$ \dfrac{\neg A\wedge B}{\neg(A\wedge B)} $\qquad
\end{center}
\begin{center}
	$(\overrightarrow{\neg\vee I})$ $ \dfrac{A\wedge\neg A}{\neg(A\vee B)} $\qquad
	$(\overrightarrow{\neg\vee E})$ $ \dfrac{\neg(A\vee B)}{A\vee\neg B} $\qquad
	$(\overleftarrow{\neg\vee I})$ $ \dfrac{B\wedge\neg B}{\neg(A\vee B)} $\qquad
	$(\overleftarrow{\neg\vee E})$ $ \dfrac{\neg(A\vee B)}{\neg A\vee B} $
\end{center}
 \begin{itemize}
\item	ND system $ \ND_{\bf K_3^\rightarrow} $ for $ \bf K_3^\rightarrow $ is obtained from $ \ND_{\bf FDE} $ by the replacement of the rules $ (\vee I_2) $, $ (\vee E) $, and $ (\neg\wedge I) $ with (EFQ), $ \overleftrightarrow{(\vee I_2)} $, $ \overrightarrow{(\vee E)} $, $ (\neg\wedge I_1) $, $ (\overrightarrow{\neg\wedge I_2}) $.
\item ND system $ \ND_{\bf K_3^\leftarrow} $ for $ \bf K_3^\leftarrow $ is obtained from $ \ND_{\bf FDE} $ by the replacement of the rules $ (\vee I_1) $, $ (\vee E) $, and $ (\neg\wedge I) $ with (EFQ), $ \overleftrightarrow{(\vee I_1)} $, $ \overrightarrow{(\vee E)} $, $ (\neg\wedge I_2) $, $ (\overleftarrow{\neg\wedge I_1}) $.
\item  ND system $ \ND_{\bf K_3^w} $ for $ \bf K_3^w $ is obtained from $ \ND_{\bf FDE} $ by the replacement of the rules $ (\vee I_1) $, $ (\vee I_2) $, $ (\vee E) $ with (EM),  $(\overleftrightarrow{\vee I_1})$, $(\overleftrightarrow{\vee I_2})$, $(\overleftrightarrow{\vee I_3})$,  and $(\overleftrightarrow{\vee E})$. 
\item ND system $ \ND_{\bf K_3^{\rightarrow 2}} $ for $ \bf K_3^{\rightarrow 2} $ is obtained from $ \ND_{\bf FDE} $ by the replacement of the rule $ (\wedge E_2) $ with (EM), $ (\wedge I_2) $, $(\overleftrightarrow{\wedge E_2})$, $(\overrightarrow{\neg\vee I})$, and $(\overrightarrow{\neg\vee E})$. 
\item ND system $ \ND_{\bf K_3^{\leftarrow 2}} $ for $ \bf K_3^{\leftarrow 2} $ is obtained from $ \ND_{\bf FDE} $ by the replacement of the rule $ (\wedge E_1) $ with (EM), $ (\wedge I_3) $, $(\overleftrightarrow{\wedge E_3})$, $(\overleftarrow{\neg\vee I})$, and $(\overleftarrow{\neg\vee E})$.
\item ND system $ \ND_{\bf PWK} $ for $ \bf PWK $ is obtained from $ \ND_{\bf FDE} $ by the replacement of the rules $ (\wedge E_1) $ and $ (\wedge E_2) $ with (EM), $ \overleftrightarrow{(\vee I_3)} $, $ (\wedge I_2) $, $(\overleftrightarrow{\wedge E_2})$, $(\overrightarrow{\neg\vee I})$,  $(\overrightarrow{\neg\vee E})$ $ (\wedge I_3) $, $(\overleftrightarrow{\wedge E_3})$, $(\overleftarrow{\neg\vee I})$, and $(\overleftarrow{\neg\vee E})$.
\end{itemize} 

Let us formulate new ND systems for the three-valued logics in question. \begin{itemize}
\item	ND system $ \ND_{\bf K_3^\rightarrow}^\prime $ (resp. $ \ND_{\bf K_3^\leftarrow}^\prime $, $ \ND_{\bf K_3^w}^\prime$) for $ \bf K_3^\rightarrow $ (resp. $ \bf K_3^\leftarrow $, $ \bf K_3^w $) is an extension of $ \ND_{\Srfde} $ (resp. $ \ND_{\Slfde} $, $ \ND_{\Sfde}^\prime $) by the rule (EFQ).
\item  ND system $ \ND_{\bf K_3^{\rightarrow 2}}^\prime $ (resp. $ \ND_{\bf K_3^{\leftarrow 2}}^\prime $, $ \ND_{\bf PWK}^\prime $) for $ \bf K_3^{\rightarrow 2} $ (resp. $ \bf K_3^{\leftarrow 2} $, $ \bf PWK$) is an extension of $ \ND_{\dSrfde} $ (resp. $ \ND_{\dSlfde} $, $ \ND_{\dSfde}^\prime $) by the rule (EM).
\end{itemize} 

\begin{proposition}
Let $ \bf L\in\{K_3^\rightarrow,K_3^\leftarrow,K_3^{\rightarrow 2},K_3^{\leftarrow 2},K_3^w,PWK \}$. 
For any set of formulas $ \Gamma $ and any formula $ A $, $ \Gamma\vdash A$ in $ \ND_{\bf L} $ iff $ \Gamma\vdash A$ in $ \ND^\prime_{\bf L} $.	
\end{proposition}
\begin{proof}
	Left for the reader.
\end{proof}


In \cite{Belikov} it is mentioned that three-valued logics $ \bf K_3^w$ and \textbf{PWK} can be formalised as extensions of $ \ND_{\Sfde}^B $ and $ \ND_{\dSfde}^B $, respectively, by the rules (EFQ) and (EM). However, the problems with $ (\vee C) $ and $ (\wedge C) $ are the same. Since our systems do not have these rules, we avoid these problems. Again, the paper \cite{KurbisPetrukhin} helps us, since there normalisation for $ \bf K_3$ and \textbf{LP} was proven (these logics extend \textbf{FDE} by (EFQ) and (EM), respectively). Let us show that the applications of (EFQ) can be restricted to propositional variables and their negations. Here are some examples, 
conjunction (the rule $ (\wedge I) $) and negated conjunction (the rule $ \overleftrightarrow{(\neg\wedge I_1^\prime)} $):
\begin{center}
	\EnableBpAbbreviations
	\AXC{$\Sigma_1$}
	\noLine\UIC{$A$}
	\AXC{$\Sigma_2$}
	\noLine\UIC{$\neg A$}
	\BIC{$B_1\wedge B_2$}
	\noLine\UIC{$\Xi$}
	\DisplayProof
	\; $\leadsto$ \;
	\AXC{$\Sigma_1$}
	\noLine\UIC{$A$}
	\AXC{$\Sigma_2$}
	\noLine\UIC{$\neg A$}
	\BIC{$B_1$}
	\AXC{$\Sigma_1$}
	\noLine\UIC{$A$}
	\AXC{$\Sigma_2$}
	\noLine\UIC{$\neg A$}
	\BIC{$B_2$}	
	\BIC{$B_1\wedge B_2$}
	\noLine\UIC{$\Xi$}
	\DisplayProof$ \qquad $
	\EnableBpAbbreviations
	\AXC{$\Sigma_1$}
	\noLine\UIC{$A$}
	\AXC{$\Sigma_2$}
	\noLine\UIC{$\neg A$}
	\BIC{$\neg(B_1\wedge B_2)$}
	\noLine\UIC{$\Xi$}
	\DisplayProof
	\; $\leadsto$ \;
	\AXC{$\Sigma_1$}
	\noLine\UIC{$A$}
	\AXC{$\Sigma_2$}
	\noLine\UIC{$\neg A$}
	\BIC{$\neg B_1$}
	\AXC{$\Sigma_1$}
	\noLine\UIC{$A$}
	\AXC{$\Sigma_2$}
	\noLine\UIC{$\neg A$}
	\BIC{$B_2$}	
	\BIC{$\neg(B_1\wedge B_2)$}
	\noLine\UIC{$\Xi$}
	\DisplayProof
\end{center}

For the case of logics with (EM) we need the following definition. 
\begin{defn}\cite[Definition 13]{KurbisPetrukhin}
	A deduction is (EM)-final if and only if there is a number of segments all of which are constituted by a sequence of formulas $C_1\ldots C_n$ such that 
	
	\noindent (i) for some $i$, $1\leq i< n$, $C_i$ is the minor premise and not the conclusion of (EM); 
	
	\noindent (ii) there are no applications of (EM) above $C_i$;
	
	\noindent (iii) for all $j$, $i\leq j<n$, $C_j$ is the minor premise of (EM) and $C_{j+1}$ is the conclusion of (EM); 
	
	\noindent (iv) $C_n$ is the conclusion of the deduction.  		
\end{defn}

\begin{lemma}\label{Lemma1}
	Any deduction in $ \bf K_3^{\rightarrow 2}$, $ \bf K_3^{\leftarrow 2} $ or \textbf{PWK} in which (EM) is applied can be transformed into one that is (EM)-final.	
\end{lemma}

\begin{proof}
	Similarly to \cite[Lemma 1]{KurbisPetrukhin}.	By repeated application of the following transformation: 
	
	\begin{center} 
		\AxiomC{$[B]$}
		\noLine
		\UnaryInfC{$\Pi_1$}
		\noLine
		\UnaryInfC{$C$}
		\AxiomC{$[\neg B]$}
		\noLine
		\UnaryInfC{$\Pi_2$}
		\noLine
		\UnaryInfC{$C$}
		\BinaryInfC{$C$}
		\noLine
		\UnaryInfC{$\Sigma$}
		\noLine
		\UnaryInfC{$D$}
		\DisplayProof\quad$\leadsto$\quad
		\AxiomC{$[B]$}
		\noLine
		\UnaryInfC{$\Pi_1$}
		\noLine
		\UnaryInfC{$C$}
		\noLine
		\UnaryInfC{$\Sigma$}
		\noLine
		\UnaryInfC{$D$}
		\AxiomC{$[\neg B]$}
		\noLine
		\UnaryInfC{$\Pi_2$}
		\noLine
		\UnaryInfC{$C$}
		\noLine
		\UnaryInfC{$\Sigma$}
		\noLine
		\UnaryInfC{$D$}
		\BinaryInfC{$D$}
		\DisplayProof
	\end{center}
	Begin with an application of (EM) lowest down in the deduction and work your way up.
\end{proof}

\begin{theorem}
	\begin{itemize}
		\item	All deductions in $ \bf K_3^\rightarrow$,  $\bf K_3^\leftarrow$,  $\bf K_3^{\rightarrow 2}$, $\bf K_3^{\leftarrow 2}$, $ \bf K_3^w$, and \textbf{PWK} can be normalised. 
		\item $ \bf K_3^\rightarrow$,  $\bf K_3^\leftarrow$,  $\bf K_3^{\rightarrow 2}$,  $\bf K_3^{\leftarrow 2}$, $ \bf K_3^w$, and \textbf{PWK} have the negation subformula property.
	\end{itemize}	
\end{theorem}
\begin{proof}
	By induction on the rank of deductions. Using the reduction steps and Lemma \ref{Lemma1}. Similarly to \cite[Theorems 7--10]{KurbisPetrukhin}.
\end{proof}	
\section{Conclusion}\label{6}
In this paper, we proved normalisation for infectious logics $ \Sfde $, $ \dSfde $, $ \bf K_3^w$, and $ \bf PWK$ as well as their non-infectious modifications, including two new logics, $ \dSrfde $ and $ \dSlfde $. Notice that all these logics are one way or another connected with Kleene's concept of regular logics. For example, $ \Sfde $ and $ \dSfde $ may be considered as four-valued versions of $ \bf K_3^w$ and \textbf{PWK}. Hence, a reasonable topic for further research is an investigation of normalisation for other Kleene-style logics, e.g. those which were formalised via ND systems in \cite{Petrukhin4}. Of course, we do not need to limit a future research to Kleene-style logics, one may try to prove normalisation for other infectious logics, e.g. for five- and six-valued versions of $ \Sfde $ and $ \dSfde $ studied in \cite{CiuniSzmucFerguson} or for logics treated in \cite{Szmuc,Szmuc21,CiuniSzmucFerguson19}.

\paragraph{Acknowledgments.} I would like to thank anonymous referees for helpful comments on the earlier version of this paper. The research presented in this paper is supported by the grant from the National Science Centre, Poland, grant number DEC-2017/25/B/HS1/01268.

\nocite{*}
\bibliographystyle{eptcs}
\bibliography{biblio}
\end{document}